\documentclass[10pt]{IEEEtran}

\IEEEoverridecommandlockouts
\usepackage{needspace}
\usepackage{theorem}
\usepackage[ruled,linesnumbered]{algorithm2e}
\usepackage{cite}
\usepackage{color}
\usepackage{mysymbol}
\makeatletter
\def\ps@headings{
\def\@oddhead{\mbox{}\scriptsize\rightmark \hfil \thepage}%
\def\@evenhead{\scriptsize\thepage \hfil \leftmark\mbox{}}%
\def\@oddfoot{}%
\def\@evenfoot{}}
\makeatother
\usepackage{amsfonts, amsmath}
\usepackage{url}
\usepackage{amsfonts}
\usepackage{amssymb}
\usepackage{acronym}
\usepackage{color}
\usepackage{graphicx}

\newcommand{\field}[1]{\mathbb{#1}}

\def\re{\mathbb{R}}

\def\1{\mathbf{1}}

\newtheorem{theorem}{Theorem}

\newtheorem{definition}{Definition}

\newtheorem{lemma}{Lemma}

\newtheorem{problem}[theorem]{Problem}

\newtheorem{remark}[theorem]{Remark}

\title{Worst-Case Scenarios for Greedy, Centrality-Based Network Protection Strategies }
\author{ Michael Zargham$^\dagger$, Victor Preciado$^\dagger$ \thanks{$^\dagger$ Victor Preciado and Michael Zargham are affiliated with the Department of Electrical and Systems Engineering at the University of Pennsylvania}}

\begin{document}
\maketitle

\begin{abstract}
The task of allocating preventative resources to a computer network in order to protect against the spread of viruses is addressed.  Virus spreading dynamics are described by a linearized SIS model and protection is framed by an optimization problem which maximizes the rate at which a virus in the network is contained given finite resources.  One approach to problems of this type involve greedy heuristics which allocate all resources to the nodes with large centrality measures. We address the worst case performance of such greedy algorithms be constructing networks for which these greedy allocations are arbitrarily inefficient.  An example application is presented in which such a worst case network might arise naturally and our results are verified numerically by leveraging recent results which allow the exact optimal solution to be computed via geometric programming.  \end{abstract}

\section{Introduction}
Given a computer network with links representing pathways through
which computer viruses can propagate, how should one distribute protection
resources to minimize the impact of a new piece of malware? The network
structure and the location of infected computers play a key role on
how quickly the malware will spread. Conversely, the location of protection
resources within the network can dramatically improve the efficiency
of protection resources aiming to contain malware spread.

Usual approaches to distribute protection resources in a network of
agents are heuristics based on network centrality measures \cite{New10}.
The main idea behind these approaches is to rank agents according
to different measures of importance based on their location in the
network and greedily distribute protection resources based on the
rank. For example, Cohen et al. \cite{cohen2003efficient} proposed
a simple vaccination strategy called \emph{acquaintance immunization
policy} in which the most connected node of a randomly selected node
is vaccinated. This strategy was proved to be much more efficient
than random vaccine allocation. Hayashi et al. \cite{HMM03} proposed
a simple heuristic called \emph{targeted immunization} consisting
on greedily choosing nodes with the highest degree (number of connections)
in scale-free graphs. Chung et at. \cite{chung2009distributing} studied
a greedy heuristic immunization strategy based on the PageRank vector
of the contact graph. Tong et al. \cite{TPTEFC10} and Giakkoupis
et al. \cite{GGTT05} proposed greedy heuristics based on immunizing
those agents that induce the highest drop in the dominant eigenvalue
of the contact graph. Recently, Prakash et al. \cite{PAITF13} proposed
several greedy heuristics to contain spreading processes in directed
networks when nodes can be partially immunized (instead of completely
removed, as assumed in previous work). These heuristics, as those
in \cite{TPTEFC10,HMM03}, are based on eigenvalue perturbation analysis.

Recently, an optimization-based approach have been developed in \cite{PZ13}-\cite{PDS13}
to solve \emph{exactly}--without relaxations or heuristics--the optimal
immunization problem in polynomial time. In particular, in \cite{PZ13,PZEJP13},
the authors proposed a convex formulation to find the optimal allocation
of protective resources in an \emph{undirected} network using semidefinite
programming (SDP). In \cite{PZEJP13Cones}, the authors solve the
optimal immunization problem in \emph{weighted and directed} networks
of \emph{nonidentical} agents using Geometric Programming (GP). Also,
in \cite{PDS13} a linear-fractional (LF) optimization program was
proposed to compute the optimal investment on disease awareness over
the nodes of a social network to contain a spreading process.

Based on the exact solution to the immunization problem developed
in \cite{PZEJP13Cones}, we propose worst-case scenarios in which
the heuristics previously proposed in the literature perform very
poorly. In our analysis, we derive the exact optimal solution for
certain directed networks and compare with those solutions obtained
using previous heuristics. Our results show how previous heuristics
can perform arbitrarily poorly in certain directed graphs.


\section{The Network Protection Problem}

The susceptible-infected-susceptible model (SIS) is a popular stochastic epidemic model first introduced by Weiss and Dishon, \cite{weiss1971asymptotic}. A discrete time variation of this model for networked populations is explored by Wang et al. in \cite{wang2009cascade}.  A continuous time version called the N-intertwined SIS model was proposed and extensively analyzed by Van Meighem et al. in \cite{van2006performance}.  An extension of the N-intertwined SIS (NiSIS) model including heterogenous agents (HeNiSIS) is presented in Preciado et al., \cite{PZEJP13Cones} and an exact solution to a family of network protection problems is also presented. The network protection problem addressed in this work is allocation of preventative resources given a fixed budget with the goal of maximizing the rate at which the epidemic is expunged. The exact solution presented in \cite{PZEJP13Cones} are leveraged to characterize the worst case behavior of common simple heuristics which greedily allocate resources within a network based on centrality measures.

\subsection{Preliminaries}
Before proceeding with the model we introduce some notation. A weighted directed graph (digraph) is defined as $G = (V, E,W)$ where $V$ is the set of n nodes, $E \subseteq V \times V$ is a set of ordered pairs of nodes indicating directed edges and edge weights $W\in \mathbb{R}_{+}^{n\times n}$ defined as weighted incidence matrix, $W_{ij} = 0$ for all $(i,j) \not\in E$.  The neighbor set of node $i$ is defined $N_i = \{ j: (i,j)\in E\}$.  For an $n \times n$ matrix $M$, the eigenvalues $\lambda_i(M)$ are ordered such that $\re(\lambda_1) \ge \re(\lambda_2) \ge \dots \ge \re (\lambda_n)$ where $\re(y)$ denotes the real part of $y\in \mathbb{C}$. 

\subsection{Virus spreading model}
The HeNiSIS model is a continuous time networked Markov process where each node in the network can be in one of two states: \textit{infected} or \textit{susceptible.}  The state is defined $X_i(t) = \{0,1\}$ for agent $i$ at time $t$ with $X_i(t)=1$ indicates the infected state. Two types of state transitions occur in this model. The probability of these transitions are defined over an infinitesimal time interval $[t, t+\Delta t).$ 
\begin{enumerate}
\item A node in the \textit{susceptible} state may become infected with a probability determined by that nodes infection rate $\beta_i$ state of its neighbors $\{X_j(t), \forall j\in N_i\}$ and the strength of the connections, $\{w_{ji}, \forall j\in N_i\}$:
\begin{align}Pr(X_i(t+\Delta t)& =1 | X_i(t) =0)=\\ &\sum_{j\in N_i} W_{ji} \beta_i X_j(t) \Delta t + o(\Delta t)\nonumber\end{align}
\item A node in the \textit{infected} state may recover from the infection based on the recovery parameter $\delta$:
\begin{equation}
Pr(X_i(t+\Delta t) = 0 | X_i(t) = 1) = \delta \Delta t + o(\Delta t)
\end{equation}
\end{enumerate} 
Analysis under this model is done using the mean field approximation.  The state variable becomes $p_i(t)$, the probability that node $i$ is infected at time $t$.  This quantity evolves according to the $n$ ordinary differential equations:
\begin{equation}
\frac{d\mathbf{p} (t)}{dt} = (BW-\delta I)\mathbf{p}(t) - P(t) BW\mathbf{p}(t) 
\end{equation}
where $\mathbf{p}(t)$ is the stacked vector of probabilities $p_i(t)$, $P(t) =$diag$(\mathbf{p}(t))$ and $B$ is the diagonal matrix with $B_{ii} = \beta_i$.  This system has stable disease free equilibrium $\mathbf{p}^*=0$.  From Proposition 1 in \cite{PZEJP13Cones}, the system stability is globally exponentially stable (with rate $\epsilon$) if $\re ( \lambda_1 (BW-\delta I)) \le -\epsilon$ for some $\epsilon>0$. 
Introducing a budget $C$ and a cost function over the protection resources $f: [\underline\beta, \bar\beta]\rightarrow \re$, the network protection problem maximizes the rate at which the virus is killed off.  Cost is incurred when decreasing the infection rate so it is assumed that $f(\beta)$ is monotonically non-increasing. 
\begin{problem} The \textbf{Network Protection Problem} is given by \label{prob}
\begin{eqnarray*}
\max_{\beta, \epsilon} && \epsilon\\
s.t. && \field{R}[\lambda_1(BW-\delta I)] \le -\epsilon\\
&& \sum_{i=1}^n f(\beta_i) \le C\\
&& \underline{\beta} \le \beta_i \le \bar \beta \qquad \forall i\in V.
\end{eqnarray*}
\end{problem}
This problem can be solved exactly via convex optimization when the function $f(\cdot)$ is a log-convex function, \cite{PZEJP13Cones}.  Knowledge of the optimal solution of problem \ref{prob} is a new development. This work proceeds to evaluate common simple heuristics in light of this knowledge.

%
%
%
%
%

\subsection{Greedy, Centrality Based Strategies}

Before addressing types of greedy heuristics, we define some additional notation.

\begin{definition}
In the vaccination problem, for any vector $u\in \field{R}^n$ over the nodes and subset of the node set $S\subseteq V$ , define the vector \[ u(S) \in \field{R}^{|S|}\]
to be the values of $u$ on the nodes in $S$.
\end{definition}

\begin{definition}
Extract the \textbf{effective objective} in Problem \ref{prob} which is induced by the epigraph form. Define 
\begin{equation}
\epsilon({\beta}) = -\field{R}[\lambda_1(BW-\delta I)] \label{ee}
\end{equation}
where $B= \hbox{diag} (\beta)$ for any feasible resource allocation $\beta$.
\end{definition}
Monotonicity and continuity of the function $\epsilon({\beta})$ guarantee that fixing any feasible $\beta$ and maximizing over $\epsilon$ always causes the constraint $ \field{R}[\lambda_1(BW-\delta I)] \le -\epsilon$ to become tight. At the optimal point $(\beta^*,\epsilon^*)$ of Problem \ref{prob} satisfies  \[ \epsilon^*=- \field{R}[\lambda_1(\hbox{diag}(\beta^*)W-\delta I)]. \]  Thus, when solving the resource allocation $\beta$, $\epsilon(\beta)$ is treated as the effective objective in Problem \ref{prob}.

\begin{definition}
Define the \textbf{efficiency} of a feasible resource allocation $\beta$ as
\begin{equation}
Q(\beta) = \frac{\epsilon(\beta)- \epsilon(\bar\beta)}{\epsilon(\beta^*)-\epsilon(\bar\beta)} \in [0,1]
\end{equation}
where $\beta^*$ is a resource allocation achieving the maximum in Problem \ref{prob}.
\end{definition}
The effective objective $\epsilon(\beta)$ and the costs functions $f(\beta_i)$ are monotonically non-increasing in the resource allocations $\beta_i$ at each node, therefore $\bar\beta$ trivially achieves the minimum over the set of feasible resource allocations $\beta$.

\begin{definition}
Let $v$ be a centrality vector.  Given a budget sufficient to completely vaccinate $k$ nodes: $C= kf(\underline\beta)$, the greedy vaccination strategy $\hat \beta_v$  is to completely vaccinate in $k$ nodes with the highest values in $v$. Define the vaccination fraction: $r=k/N$ where $N$ is the the total number of nodes.
\end{definition} 

Common centrality measures used for heuristics are degree and eigenvector centrality, \cite{HMM03}. Page rank centrality is used as in place of eigenvector centrality in the case of general digraphs, \cite{chung2009distributing}.  While Page rank depends on a parameter $\alpha$, we drop the $\alpha$ from our notation because our results hold for the whole family of Page rank vectors generated by non-trivial choices of $\alpha\in (0,1)$. \\

\section{Analytical Results}

\begin{figure}[b]
\includegraphics[width=.7\columnwidth]{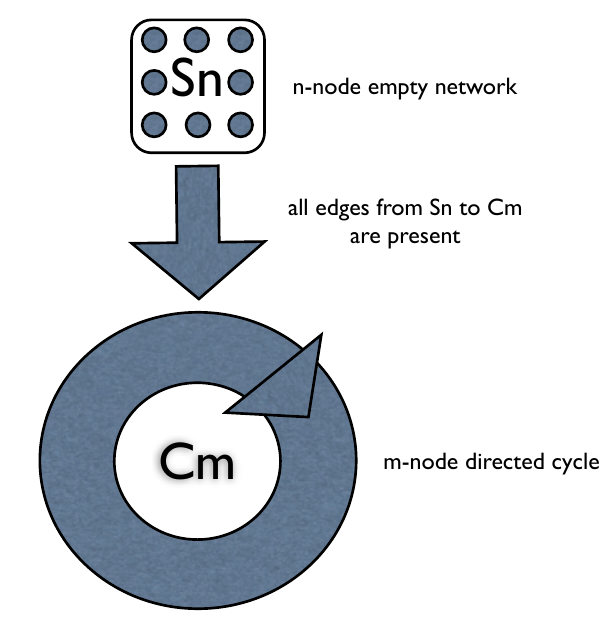}
\centering
\caption{\label{wheel}  We construct the network $G$ to prove theorem \ref{counterex}.}
\end{figure}
\begin{theorem}\label{counterex}
Given a budget $C$, there exists a network $G$ satisfying 
\[Q(\hat\beta_{DEG})=Q(\hat\beta_{PR}) = 0\]
where  $r\in(0,1)$ is the fraction of nodes that can be vaccinated.
\end{theorem}

\begin{proof}
Construct the network $G=\{V,E\}$ as follows, the node set is partitioned $V = C_m \cup S_n$ where $|S_n|=n$, $|C_m|=m$ and $N=m+n$. Choose any $n$ and $m$ satisfying
\begin{equation}
2<m\le n\left(\frac{1}{r}-1\right) \hbox{  and  }
C \le n f(\underline\beta).
\end{equation}

As shown in Fig. \ref{wheel}, define the subgraph $C_m$ as an $m$ node directed cycle, $S_n$ as an $n$ node empty network and the there are edges from all nodes $i\in S_n$ to all nodes $j\in C_m$.  Formally, the edge set is given by
\begin{equation}
(i,j) \in E \hbox{ if any of }\left\{ \begin{array}{l} i\in S_n, j\in C_m\\ i, j=i+1\in C_m\\ i=m+n, j=n+1\in C_m \end{array}\right.
\end{equation}
and in all other cases $(i,j)\not \in E.$ Let us consider the object function in the case of $G$, $\epsilon(\beta) =$
\begin{equation}
 -\field{R}\lam_1\left(\delta I-\left[\begin{array}{cc} \hbox{diag}\beta(S_n) & 0 \\ 0 &\hbox{diag}\beta(C_m) \end{array} \right] \left[ \begin{array}{cc} 0 & \mathbf{11'} \\ 0 & U\end{array}\right]\right)
\end{equation}
where $U$ is the adjacency matrix of the directed cycle. Block multiplication yields $\epsilon(\beta) =$
\begin{equation}
- \field{R}\lam_1\left(\delta I- \left[ \begin{array}{cc} 0 & \hbox{diag}\beta(C_m)(\mathbf{11'})\\ 0 & \hbox{diag}\beta(C_m)U\end{array}\right]\right) 
\end{equation}
which due to the block diagonal structure simplifies to
\begin{equation}
\epsilon(\beta) = \field{R}\lam_1\left(\hbox{diag}\beta(C_m)U\right)-\delta
\end{equation}
which tells us that the optimal budget allocation is over the nodes in $C_m$,
\begin{equation}\beta^* = [\bar\beta(S_n), \beta^*(C_m)].\label{dep}\end{equation}
Consider the out-degree\footnote{Total degree centrality may be substituted for out degree by imposing the condition $m>n+2$. Such an $(n,m)$ pair exists and is chosen for demonstrations in section \ref{demos}.} vector of G:
\begin{equation}
DEG(i) = \left\{ \begin{array}{ll} 1, & \forall i\in C_m\\ m, & \forall i\in S_n \end{array}\right.
\end{equation}
By our choice of $m$ and $n$, $rN<n$ and the greedy heuristic solution $\hat \beta_{DEG}$, satisfies 
\begin{equation}
i \in C_m \implies \hat \beta_{DEG} (i) = \bar\beta_i.
\end{equation}
From \eqref{dep} infection rates for nodes in $S_n$ have no impact on $\epsilon$, thus $\epsilon(\hat \beta_{DEG}) = \epsilon(\bar \beta)< \epsilon(\beta^*).$
For the case of $\hat\beta_{PR}$, it is necessary to compute the Page rank vector for $G$.  The Page rank vector is the dominant eigenvector of the stochastic matrix
\begin{equation}
\frac{\alpha}{N} \mathbf{11'} + (1-\alpha) \bar W \label{PRdef}
\end{equation}
where $\alpha\in (0,1)$ is the teleportation parameter and $\bar W$ is a random walk matrix on $G$ given by
 \begin{equation}
 \bar W_{ij}= \left\{\begin{array}{cc} \frac{1}{1+DEG_i}, & (i,i) \hbox{ and } (i,j) \in E \\ 0, & \hbox{otherwise}\end{array}\right.. \label{RWdef}
 \end{equation}
  Due to the structure (symmetry within $S_n$ and $C_m$) of $G$ and the fact that the dominant eigenvalue of a stochastic matrix is 1, we can write the block matrix eigenvector equation
\begin{equation}
P_\alpha\left[ \begin{array} {c} \chi \mathbf{1}\\ \mathbf{1}\end{array} \right]
=\left[ \begin{array}{c} \chi \mathbf{1}\\ \mathbf{1} \end{array} \right] \label{FP}
\end{equation}  
where
\begin{equation}
P_\alpha =
\left[ \begin{array}{cc} \frac{\alpha}{N} \mathbf{11'}+{(1-\alpha)} I&\frac{\alpha}{N} \mathbf{11'} +\frac{(1-\alpha)}{m+2}\mathbf{11'}   \\ \frac{\alpha}{N} \mathbf{11'}& \frac{\alpha}{N} \mathbf{11'} + \frac{(1-\alpha)}{m+2}(U+I)\\ \end{array} \right] \label{EV}
\end{equation}  
where $U$ is the adjacency matrix for $C_m$, which satisfies $U'\mathbf{1}= \mathbf{1}$ . 
Simplifying equation defined by \eqref{EV} and \eqref{FP} by multiplying through by each instance of $\mathbf{1}$, we have
\begin{eqnarray}
\frac{\alpha}{N}(\chi n+m) +(1-\alpha)\chi +(1-\alpha)\frac{m}{m+2}&=& \chi \label{eq1}\\
\frac{\alpha}{N}(\chi n+m) +(1-\alpha)\frac{2}{m+2}&=& 1\label{eq2}.
\end{eqnarray}
Subtracting \eqref{eq2} from \eqref{eq1} and simplifying, we have
\begin{equation}
\chi = \frac{1}{\alpha}\left(1+\frac{(1-\alpha)(m-2)}{m+2}\right).
\end{equation}
We have selected $m>2$, so it is guaranteed that $\chi>1$ for all $\alpha\in(0,1)$.  From equation \eqref{FP} we have the probability vector
\begin{equation}
PR(i) = \left\{ \begin{array}{ll} \frac{1}{n\chi +m} & \forall i\in C_m\\  \frac{\chi}{n\chi+m} & \forall i\in S_n \end{array}\right. \label{scaling}
\end{equation}
with $\chi>1$, guaranteeing that nodes $i\in S_n$ always have greater Page rank centrality than nodes in $C_m$. Using the same argument as in the degree centrality case, we have $\epsilon(\hat\beta_{PR}) = \epsilon(\bar \beta)< \epsilon(\beta^*)$. \end{proof}

\begin{remark}
The proof of  Theorem \ref{counterex} makes use of a constructive example for the centrality measures which identify nodes which are the most likely to become infected: (a) out degree and (b) Page rank with a random walk defined as moving up the edges.  If one uses centrality measures which identify nodes which would be the most potent seeds such as (c) in degree or (d) Page rank computed using a random walk that flows down the edges, one can construct an alternative $G$ by simply reversing the direction of the edges from $S_n$ to $C_m$. Using this alternative network, one can reproduce Theorem \ref{counterex} for (c) and (d).
\end{remark}

Theorem \ref{counterex} tells us that for a general digraph, the greedy allocation strategy can be arbitrarily bad.  However, common network resource allocations take the basic assumption that the graph is strongly connected. Since the proof constructs a graph which is weakly, but not strongly connected we develop a related theorem for strongly connected digraphs.

\begin{lemma}\label{lemma}
There is a 1 parameter family of strongly connected digraph $G'_\gamma=\{V, E \cup E_\gamma\}$ for $\gamma>0$ with the node set $V=S_n \cup C_m$, for which the greedy centrality based strategies $\hat \beta_{DEG}$ and $\hat \beta_{PR}$ allocate resources only in $S_n$.
\end{lemma}

\begin{proof}
Consider the $N$ node network $G=\{V,E\}$ presented in the proof of theorem \ref{counterex}. We construct a new network $G'_\gamma = \{V, E \cup E_\gamma\}$ where $E_\gamma$ is all to all and these edges are assigned positive weight $\gamma>0$.

With the addition of the all to all edge set, it is trivial to observe that the out degree centrality vector becomes
\begin{equation}
DEG(i) = \left\{ \begin{array}{ll} 1+\gamma N & \forall i\in C_m\\ m+ \gamma N & \forall i\in S_n \end{array}\right.
\end{equation}
By our choice of $m$ and $n$, $rN<n$ and the greedy heuristic solution $\hat \beta_{DEG}$, satisfies 
\begin{equation}
i \in C_m \implies \hat \beta_{DEG} (i) = \bar\beta_i.
\end{equation}
therefore the resource allocation is entirely on $S_n$.\\

Now we consider the case of Page rank centrality according to the definition in equation \eqref{PRdef} for the random walk $\bar W$ on $G_\gamma'$ computed according to \eqref{RWdef}.  The fixed point equation is given by \eqref{FP} where
\begin{equation}
P_\alpha=\frac{\alpha}{N}\mathbf{11'}+(1-\alpha)\left[ \begin{array}{cc} \frac{I+\gamma \mathbf{11'}}{1+N\gamma}&\frac{(1+\gamma)\mathbf{11'}}{m+2+N\gamma}   \\ \frac{\gamma \mathbf{11'}}{1+N\gamma}& \frac{(I+U+\gamma\mathbf{11'})}{m+2+N\gamma}\\ \end{array} \right] \label{gamP}
\end{equation}  
Simplifying equation \eqref{FP} with $P_\alpha$ defined by \eqref{gamP} by multiplying through by each instance of the vector $\mathbf{1}$ and consolidating terms, we have

\begin{equation}
\begin{array}{l}\frac{\alpha n \chi + (1-\alpha)m }{N}  +\frac{(1-\alpha)(1+\gamma n)\chi}{1+N\gamma}+\frac{(1-\alpha)(1+\gamma m)}{m+2+N\gamma}\end{array}= \chi \label{Geq1}
\end{equation}
\begin{equation}
\begin{array}{l}\frac{\alpha n \chi + (1-\alpha)m }{N}  +\frac{(1-\alpha)(2+m\gamma)}{m+2+N\gamma}\end{array}= 1\label{Geq2}.
\end{equation}
Subtracting \eqref{Geq2} from \eqref{Geq1} many terms cancel and allowing us to simplify and solve for $\chi$:
\begin{equation}
\chi = \frac{1+N\gamma}{\alpha+N\gamma}\left(1+\frac{(1-\alpha)(m-2)}{m+2+N\gamma}\right).
\end{equation}
The leading term in the product $\frac{1+N\gamma}{\alpha+N\gamma}>1$ and we have selected $m>2$, so $\chi>1$ is guaranteed for all $\gamma>0$.   From \eqref{scaling}, the Page rank of any node in $S_n$ is a factor of $\chi>1$ times that of any node in $C_m$, guaranteeing that nodes $i\in S_n$ always have greater Page rank centrality than nodes in $C_m$. Following the same argument as for the $DEG$ based allocation, the $PR$ greedy resource allocation is entirely on the set $S_n$.
\end{proof}

Lemma \ref{lemma} characterizes the greedy resource allocation strategies on the family of digraphs $G_\gamma'$. Specifically, the infection rate profiles achieved $\hat\beta_{DEG}$ and $\hat\beta_{PR}$ do not depend on $\gamma$.  For any $\gamma$ nodes in $S_n$ are fully immunized until all resources are expended. 

\begin{theorem}\label{digraph}
There exists a strongly connected digraph for which the centrality based greedy solutions $\hat\beta_v$ are arbitrarily inefficient solutions to Problem 1, in the sense that
\[Q(\hat\beta_v)=\frac{\epsilon(\hat\beta_v)- \epsilon(\bar\beta)}{\epsilon(\beta^*)-\epsilon(\bar\beta)} \le \Gamma  \]
for any $\Gamma>0$, when the centrality vector $v=DEG,$ or $PR$.
\end{theorem}
\begin{proof} 
Define a $\gamma$-parameterized version of the effective objective in \eqref{ee} as
\begin{equation}
\epsilon({\beta}; \gamma) = -\field{R}[\lambda_1(\hbox{diag}(\beta)(W+\gamma\mathbf{11'})-\delta I)] \label{eeg}
\end{equation}
where $W$ is the adjacency matrix of $G$ so we have explicitly written the effective objective of a infection rate profile $\beta$ on the network $G_\gamma'$ as defined in Lemma 1.  For any fixed $\beta$ the function $\epsilon_\beta(\gamma) = \epsilon(\beta; \gamma)$ is a scalar mapping $\epsilon_\beta: \re_+ \rightarrow \re$. The eigenvalues of a matrix are continuous functions in the matrix elements because they are the roots of the characteristic equation, \cite{Z65}. Thus the mapping $\epsilon_\beta(\gamma)$ is continuous in $\gamma$ because it is a composition of continuous functions.

Define a $\gamma$-parameterized efficiency function using the $\gamma$-parameterized effective objective
\begin{equation}
Q(\beta; \gamma) = \frac{\epsilon_\beta(\gamma)- \epsilon_{\bar\beta}(\gamma)}{\epsilon(\beta^*;\gamma)-\epsilon_{\bar\beta}(\gamma)} \in [0,1].
\end{equation}
Observe that the optimal resource allocation $\beta^*$ depends on $\gamma$ and satisfies $\epsilon(\beta^*) > \epsilon(\bar\beta)$ for all networks guaranteeing that $\epsilon(\beta^*; \gamma) > \epsilon(\bar\beta; \gamma)=\epsilon_{\bar\beta}(\gamma)$ for all $\gamma\in \re_+$.  Having established that the denominator in the quotient is strictly positive, conclude that the function $Q(\beta;\gamma)$ is continuous in $\gamma$ because it can be constructed as sums and products of continuous functions.

From Theorem \ref{counterex}, the efficiency of $\hat\beta$ selected according degree or Page rank yields zero efficiency on network $G$ which can be rewritten in terms of the $\gamma$-parametrized efficiency as
\begin{equation}
Q(\hat\beta_{DEG}; \gamma=0)=Q(\hat\beta_{PR}; \gamma=0) = 0.
\end{equation}
Leveraging the continuity of $Q(\beta; \gamma)$, there exists a $\gamma$ such that $Q(\hat\beta_{DEG}; \gamma), Q(\hat\beta_{PR}; \gamma)\le \Gamma$ for any $\Gamma>0$.
\end{proof}

Theorem \ref{digraph} shows that having a strongly connected digraph does not remove the possibility that greedy centrality based networks will perform very poorly.  The digraph family constructed in the proof is only one method to produce a worst-case digraph.  Other worst case digraphs may be constructed but few are so easily analyzed.

\begin{figure}
\includegraphics[width=\columnwidth]{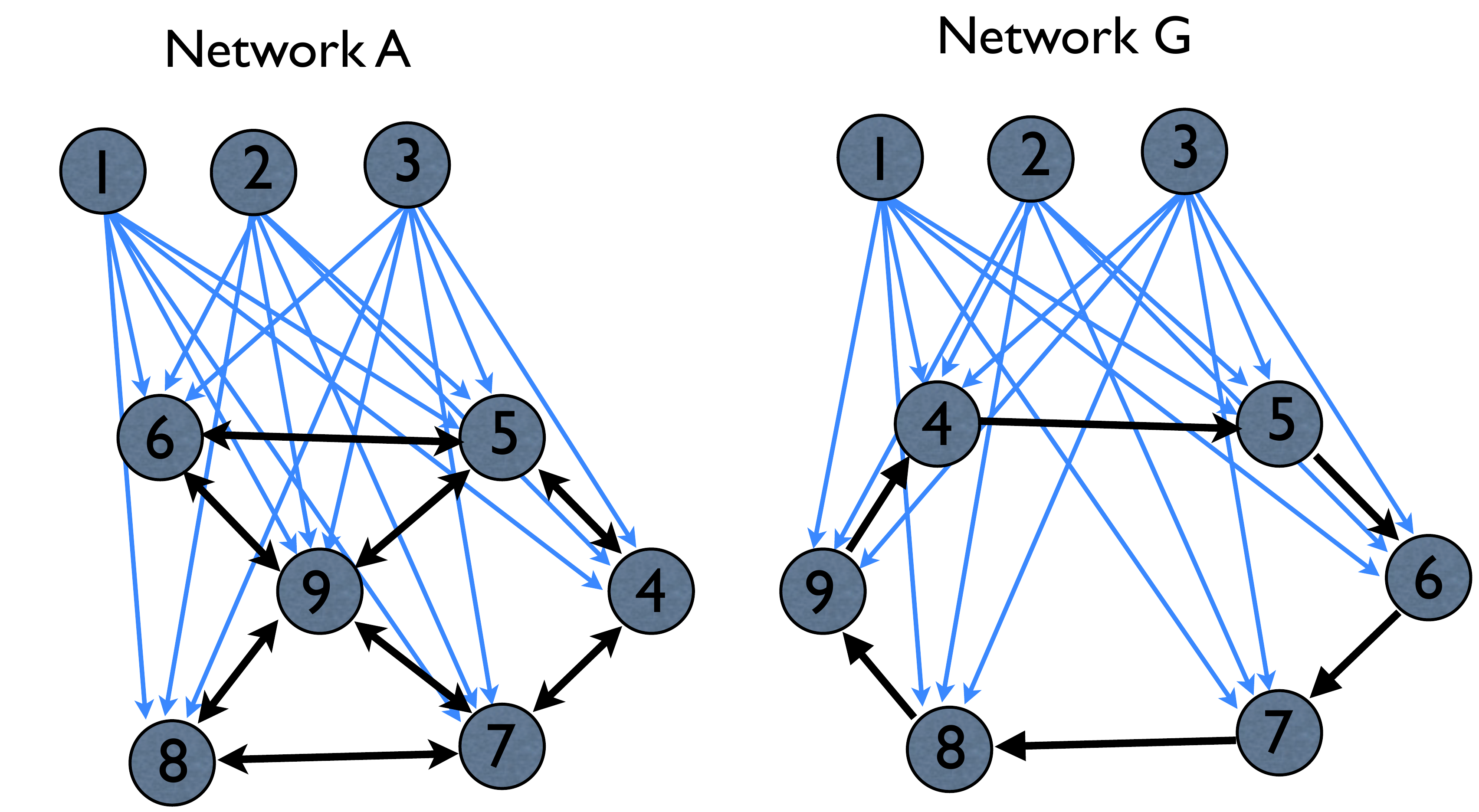}
\centering
\caption{\label{nets} Network G with vertices $S_3=\{1,2,3\}$ and $C_6=\{4,5,\dots,9\}$ satisfies the conditions for the counter example network defined in Theorem \ref{counterex}.  In Network A the subgraph on $C_6$ is relaxed to be less structured for demonstration purposes.}
\end{figure}
\begin{figure}
\includegraphics[width=\columnwidth]{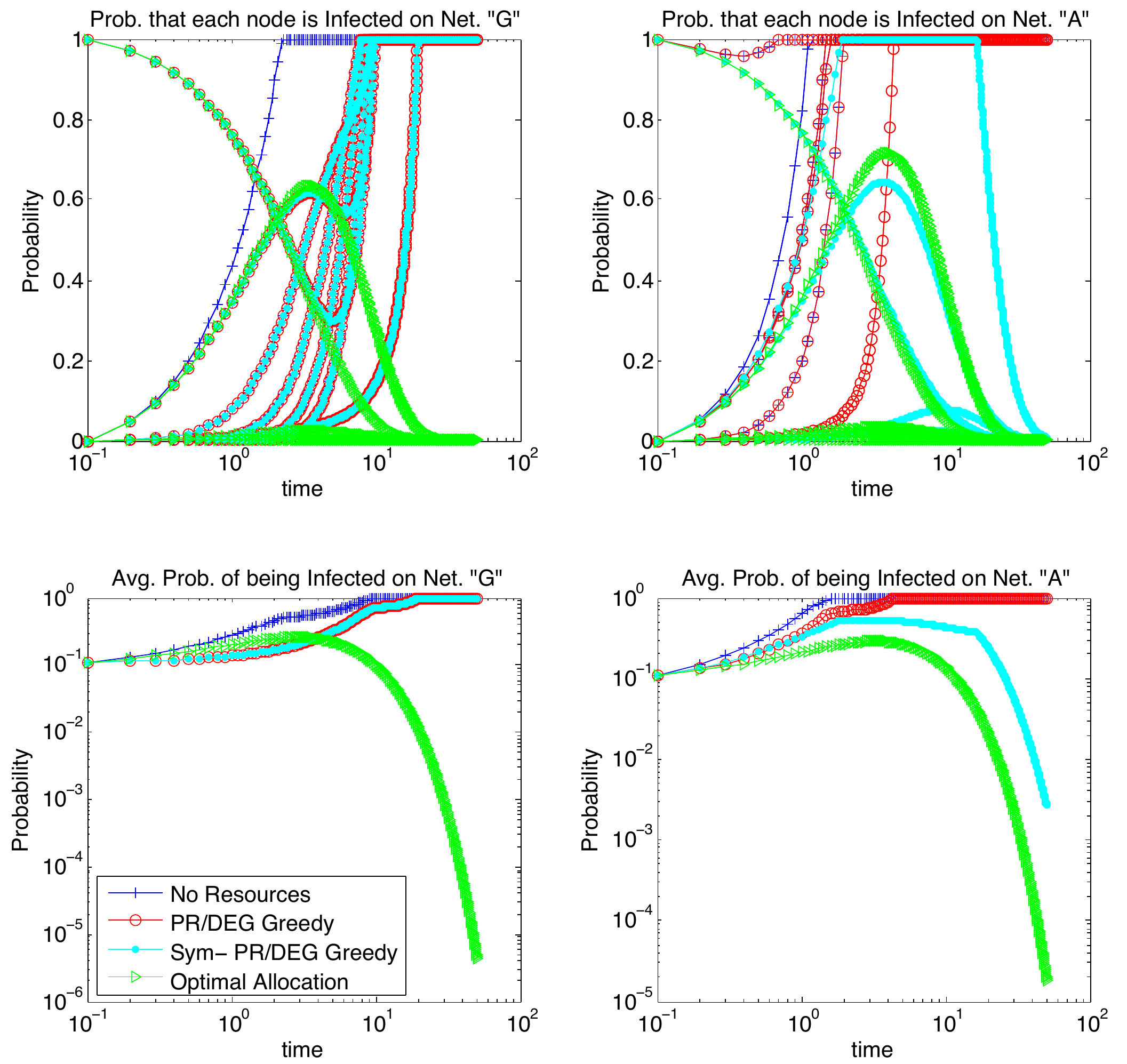}
\centering
\caption{\label{plots} The HeNiSIS dynamics are considered on the networks A and G when one node initially has an infection with probability 1. (Top, Left) In the counter example network G, all greedy algorithms fail to prevent the outbreak while the optimal allocation protects the network. (Bottom, left) The rate at which the virus is expunged by the optimal solution is exponential. (Top, Right) In network A, the symmetric centralities measures and the optimal allocation eventually eliminate the virus. (Bottom, Right) The rate at which the optimal allocation eliminates the virus is faster.}
\end{figure}
\section{Computational Results}\label{demos}
We consider a simple application in which such a worst case network might arise naturally. Nodes are computers belonging to individuals in a work environment. Edges indicate access to files on another persons computer. $C_m$ consists of a group of workers and $S_n$ a group of administrators who can access files on all works computers. Workers have limited access to each others computers but do not have access to files on the administrator's computers. We assume the virus may spread when an uninfected computer accesses an infected computer. Protection resources take the form of antivirus software with updates on a variable time interval,  software updated more frequently providing a smaller infection rate $\beta$ but updates incurring a greater cost $f(\beta)$. The cost function 
\begin{figure*}
\includegraphics[width=\textwidth]{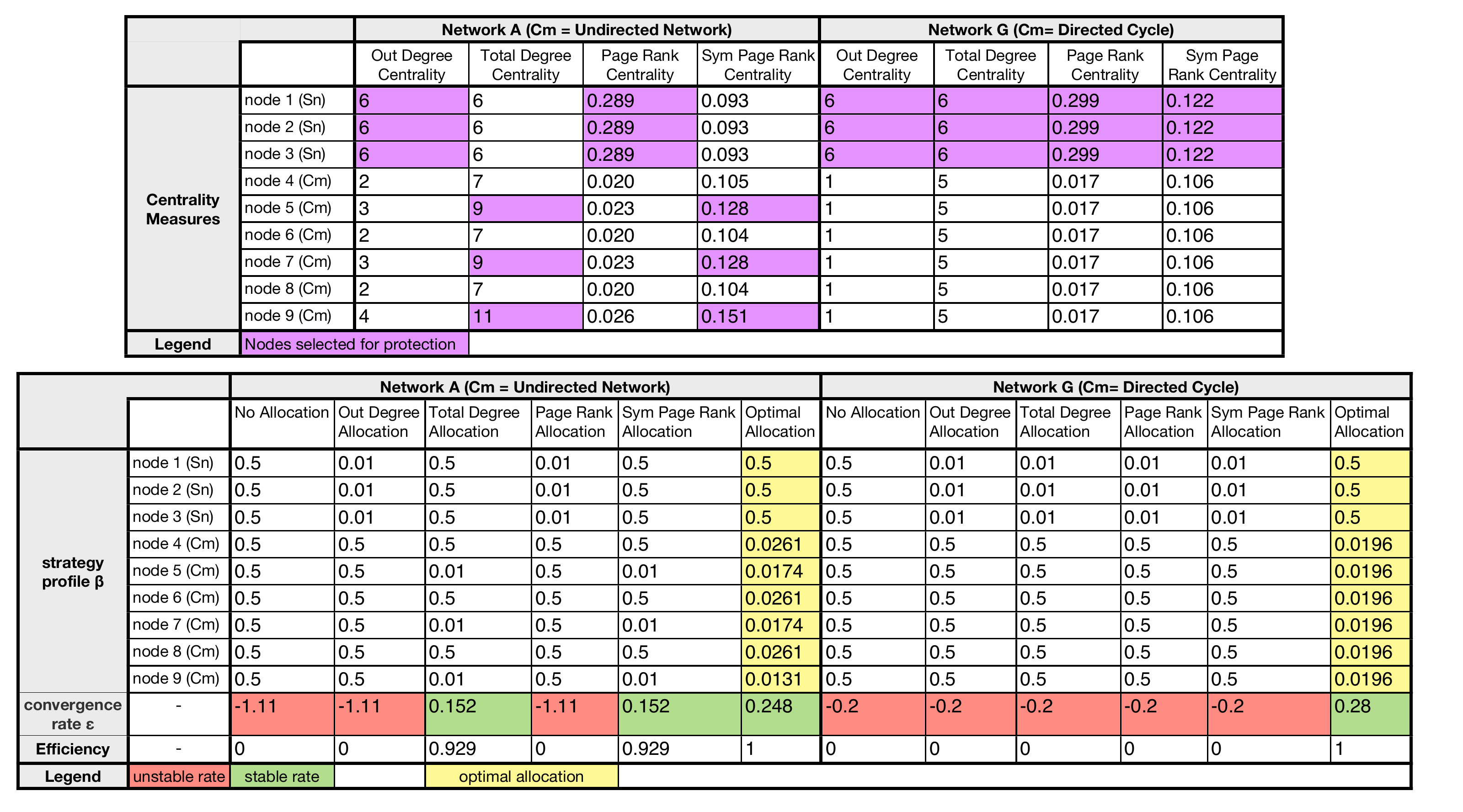}
\centering
\caption{\label{data} (Top) A variety of centrality measures are used as the basis for greedy algorithms, these measures are reported for the Networks A and G. (Bottom) The allocation strategies tested are detailed, their exponential convergence rate bounds $\epsilon$ and their efficiencies are reported for comparison purposes.  For the case of the counter example network G, none of the greedy type algorithms yield a stable convergence rate.
 }
\end{figure*}
\begin{equation}
f(\beta_i) = \frac{\underline\beta(\frac{\bar \beta}{\beta_i}-1)}{\bar\beta-\underline \beta }
\end{equation}
 is chosen to satisfy $f(\bar\beta)=0$, $f(\underline\beta)=1$ and $f(\beta) \propto 1/\beta$.  This allows us to choose capacity $C$ equal to the number of nodes we wish to be able to allocate maximum protection. In our example the infection rate with outdated anti-virus software is $\bar\beta = .5$ while the maximum update rate achieves an infection rate of $\underline\beta = .01$.  Choosing a budget of $C=3$ for a network with $n=3$ and $m=6$ (such as in G or A shown in Fig. \ref{nets}), the fraction of nodes that can be maximally protected is $r=1/3$.  An infected machine has recovery rate $\delta= 0.3$, based on curative resources which are uniformly available.
 
In the example, four heuristic algorithms based on greedily allocating resources with respect to centrality measures are considered.  The centrality measures are out degree, total degree, Page rank with $\alpha=.1$ and symmetrized Page rank with $\alpha=.1$. Symmetrized Page rank is computed by allowing the random walk move over a directed edge in either direction. The worst case networks are products of extreme asymmetry between $C_m$ and $S_n$, the symmetric centrality measure show that even symmetric centrality measure don't overcome the potential for arbitrarily poor behavior. 

In Fig. \ref{data} the top table shows all of the centrality vectors for the example problem in the networks $A$ and $G$. The network $G$ is the network constructed in our analytical proofs. The network $A$ is an example of a less structured employee collaboration network which we include to demonstrate two points: (i) our constructed network G is not unique and (ii) symmetrizing heuristics are less fragile than heuristics that respect edge direction. 

In $G$ and $A$ the out degree and Page rank heuristics allocate all resources to the admins, $S_n$. This is ineffective because even though the admins are the most likely to become infected the worker group, $C_m$ cannot access their files and become infected. Fig. \ref{data} (bottom) shows the infection rate profiles generated by the various heuristics and the optimal solution. A strategy is ineffective if the convergence rate epsilon is negative because this corresponds to unstable dynamics and the probability of infection becoming one for all machines. Figure \ref{plots} demonstrates the dynamic under each of the strategy profiles stated in Figure \ref{data}.

%
%

\section{conclusion}
We have proven that for common centrality measures there exist networks for which greedy allocation of protection resources is completely ineffective.  Furthermore, these worst case networks are not completely unreasonable pathological cases. An application in which this network structure could arise naturally is presented.

In practice, if the information and computational power to solve the optimization via GP are available, this method should always be used. Restrictions on computational power or complete information may still lead to use of heuristics.  Degree can be computed locally and Page rank can be approximated iteratively. When using these heuristics, we suggest using symmetric variants; while we showed that even the symmetric heuristics can be arbitrarily inefficient, we found the symmetric variants to be more robust.
%
%
%
\bibliographystyle{unsrt}
\bibliography{ViralSpread}
\end{document}